\documentclass[11pt,letterpaper]{article}

\usepackage{fullpage}

\usepackage{amsmath}
\usepackage{amsfonts}
\usepackage{amssymb}
\usepackage{amsthm}
\usepackage{amsbsy}

\usepackage{graphicx}
\usepackage{subcaption}
\usepackage{libertine}
\usepackage{bbm}

\usepackage{hyperref}
\usepackage[svgnames]{xcolor}
\usepackage[capitalise,nameinlink]{cleveref}
\hypersetup{colorlinks={true},linkcolor={DarkBlue},citecolor=[named]{DarkGreen}}

\usepackage{natbib}

\newtheorem{theorem}{Theorem}
\newtheorem{lemma}[theorem]{Lemma}

\newtheorem{corollary}[theorem]{Corollary}

\theoremstyle{definition}

\newtheorem{example}{Example}

\newcommand{\bb}{\mathbf{b}}

\newcommand{\yy}{\mathbf{y}}

\newcommand{\OO}{\mathbf{O}}
\newcommand{\XX}{\mathbf{X}}
\newcommand{\YY}{\mathbf{Y}}

\newcommand{\LW}{\text{\normalfont LW}}

\newcommand{\lpoa}{\text{\normalfont LPoA}}
\newcommand{\lpos}{\text{\normalfont LPoS}}

\newcommand{\calG}{\mathcal{G}}

\newcommand{\calC}{\mathcal{S}}

\newcommand{\eq}{\text{EQ}}

\newcommand{\remove}[1]{}

\title{\bf Simple combinatorial auctions with budget constraints}

\author{Alexandros A. Voudouris \\ School of Computer Science and Electronic Engineering \\ University of Essex}

\date{}

\begin{document}

\maketitle

\begin{abstract}
We study the efficiency of simple combinatorial auctions for the allocation of a set of items to a set of agents, with private subadditive valuation functions and budget constraints. The class we consider includes all auctions that  allocate each item independently to the agent that submits the highest bid for it, and requests a payment that depends on the bids of all agents only for this item. Two well-known examples of this class are the simultaneous first and second price auctions. We focus on the pure equilibria of the induced strategic games, and using the liquid welfare as our efficiency benchmark, we show an upper bound of $2$ on the price of anarchy for any auction in this class, as well as a tight corresponding lower bound on the price of stability for all auctions whose payment rules are convex combinations of the bids. This implies a tight bound of $2$ on the price of stability of the well-known simultaneous first and second price auctions, which are members of the class. Additionally, we show lower bounds for the whole class, for more complex auctions (like VCG), and for settings where the budgets are assumed to be common knowledge rather than private information. 
\end{abstract}

\section{Introduction}\label{sec:intro}
We study the efficiency of auctions for the allocation of a set of discrete items to a set of {\em budget-constrained} agents with combinatorial preferences over the items, expressed via {\em valuation functions}. Such {\em combinatorial auctions} have been studied extensively in settings without budget constraints, through the lens of {\em mechanism design}~\citep{Dobzinski2011impossibility,Dobzinski2016submodular,Feldman2015fixed} and {\em equilibrium analysis}~\citep{Christodoulou2016bayesian,Bhawalkar2011bidding,Feldman2013simultaneous}, under various assumptions about the structure of the valuation functions~\citep{Lehmann2006classes}. 
While designing approximately optimal and truthful auctions is important, such an approach is usually restricted by impossibility results and also leads to auctions with complicated allocation and payment rules. 
In this paper we follow the alternative approach of using much simpler auctions, like the {\em simultaneous first and second price auctions}, where every agent submits a scalar bid per item, each item is allocated to the agent that submits the highest bid for it, and the winner of every item either pays her own bid (first price) or the second highest bid (second price). However, such auctions are non-truthful and naturally induce strategic games among the agents who typically act as utility-maximizing players and strategize over their bids. The objective then is to quantify the efficiency loss at equilibrium due to the strategic behavior of the agents, by bounding the {\em price of anarchy} \citep{KP99} and the {\em price of stability} \citep{Anshelevich2008pos} in terms of some social efficiency benchmark. 

Most of the related literature on auctions makes two main assumptions about the characteristics of the agents and the efficiency benchmark. First, the agents are assumed to have {\em no budget constraints}, and are therefore able to afford any payment as long as they get non-negative utility. However, such an assumption is clearly non-realistic. Even though an agent might have high value for a set of items, she may not have the necessary liquidity to actually pay for them. Therefore, it is more natural to instead consider settings where the agents may have budget constraints that limit their ability to pay. Second, social efficiency is measured by the {\em social welfare} benchmark, defined as the total value of the agents for the items they are allocated. While this might be true if we assume that the agents do not have budget constraints, when budgets do exist, it is not hard to see that the social welfare achieved at equilibrium can be arbitrarily far away from the maximum possible; consider the following example with just two agents and one item.

\begin{example}\label{ex:social}
Assume that there is one item, and two agents with values $(\lambda > 2, 2)$ and budgets $(1,2)$. Since the second agent has more budget than the first agent ($2$ vs. $1$), she can always outbid the first agent and get the item. This leads to an allocation with social welfare $2$. However, the optimal allocation is to give the item to the first agent for a social welfare of $\lambda$. Therefore, there is a huge gap between the welfare at equilibrium and that of the optimal allocation if $\lambda$ is large.
\qed
\end{example}

\noindent 
Therefore, the social welfare is clearly not the correct benchmark in settings with budget constraints. Instead, we use the {\em liquid welfare} benchmark, which takes the budgets of the agents into account. Similarly to the social welfare, the liquid welfare of an allocation is defined as the total value of all agents, but with the value of each agent capped off by her budget (i.e., we take the minimum between the agent's value and her budget). Observe that for this benchmark the equilibrium allocation of Example~\ref{ex:social} is now the optimal one. 
The liquid welfare benchmark was first proposed independently by \citet{Dobzinski2014liquid} and \citet{ST13}, who called it {\em effective welfare}. Dobzinski and Paes Leme showed that the well-known VCG auction~\citep{V61,C71,G73} is no longer truthful, and designed truthful mechanisms that approximate the optimal liquid welfare for settings with divisible items. On the other hand, Syrgkanis and Tardos compared the social welfare at equilibrium to the optimal liquid welfare, for games induced by a plethora of auctions, ranging from single-item and multi-unit auctions to combinatorial ones.

The price of anarchy in terms of the liquid welfare, termed {\em liquid price of anarchy}, has been studied in a series of papers by now. To start with, the papers by \citet{CV16} and \citet{Christodoulou2016proportional} focused on the proportional allocation mechanism, used for the allocation of divisible resources, and were the first to formally bound the liquid price of anarchy. Following their work, \citet{CV18} characterized the structure of worst-case pure equilibria and proved tight bounds on the liquid price of anarchy for almost all divisible resource allocation mechanisms. In a similar spirit, \citet{V19} showed tight bounds on the pure liquid price of anarchy and stability for ad auctions, including the generalized second price auction and VCG. Our work is mostly related to that of \citet{Azar2017lpoa} who showed a constant bound on the liquid price of anarchy of simultaneous first and second price auctions for Bayes-Nash equilibria when the agents have additive valuation functions, and a tight bound of $2$ on the liquid price of anarchy for pure equilibria when the agents have fractionally-subadditive valuations. Among other results, in this paper we extend the second main result of \citet{Azar2017lpoa} to the class of subadditive valuation functions, which contains that of fractionally-subadditive, and further show that the bound of $2$ in fact holds for the liquid price of stability, and for many more simple combinatorial auctions.

\subsection{Our contribution}
We consider a combinatorial setting, where a set of $m$ discrete items are to be allocated to a set of $n$ budget-constrained agents with {\em subadditive} valuation functions. For the allocation of the items, we consider simple combinatorial auctions, which work as follows: every agent submits a bid per item, receives every item for which she submits the highest bid, and also pays an amount that is a function of the bids of all agents for this item only; such auctions were considered before by \citet{Bhawalkar2012simultaneous} in settings without budget constraints. We focus on the {\em pure} equilibria of the strategic games induced by such auctions, and study their social efficiency, by bounding the liquid price of anarchy and stability; that is, we bound the worst-case ratio between the liquid welfare achieved by any possible allocation and the liquid welfare in the worst and the best equilibrium, respectively.

By carefully exploiting techniques developed by \citet{Bhawalkar2011bidding} and \citet{Feldman2013simultaneous}, we show an upper bound of $2$ on the liquid price of anarchy of all simple combinatorial auctions whose payment rules satisfy some mild assumptions; see Section~\ref{sec:upper}. We complement this upper bound with a corresponding lower bound of $2$ on the liquid price of stability for auctions whose payment rules are convex combinations of the bids, and a lower bound of $2-1/n-(n-1)/m$ for all simple auctions with $n$ players and $m$ items. As a corollary, we obtain that the liquid price of stability of simultaneous first and second price auctions with subadditive agents is exactly $2$, thus extending the corresponding bound of \citet{Azar2017lpoa} for fractionally-subadditive agents; these results can be found in Section~\ref{sec:lower}.
We also consider more complex auctions, which take into account the bids for all items when deciding the allocation and the payments, and show that even the most prominent such auction, VCG, has liquid price of anarchy at least $2$; see Section~\ref{sec:vcg}. Finally, in Section~\ref{sec:known} we consider the case where the budgets are common knowledge, instead of private information. While one would expect that this extra power could lead to an almost fully efficient simple auction, we actually show that this is not true, by presenting a lower bound of $4/3 - 2/(3m)$ on the liquid price of anarchy of any simple auction. We conclude the paper by discussing possible directions for future research in Section~\ref{sec:open}. 

\subsection{Other related work}
The related literature on auction theory is extensive. We refer the interested reader to the survey of \citet{survey} and the chapter of \citet{combinatorial-chapter} for a broad introduction to combinatorial auctions, as well as to the survey of \citet{poa-survey} for an overview of the work on the price of anarchy in auctions.

There has been a lot of recent work on auctions with budget-constrained agents. \citet{LX15,LX17} extended and also generalized the work of \citet{Dobzinski2014liquid}. \citet{Fotakis2019bridge} showed how truthful auctions which approximate the social welfare in submodular valuation settings without budget constraints can be adapted to settings with budget constraints, so that they remain truthful and also now achieve almost the same approximation for the liquid welfare. In a slightly different context, \citet{Branzei2017truthful} viewed the liquid welfare as an upper bound on the maximum possible revenue and designed truthful mechanisms for revenue maximization, while \citet{Branzei2019dynamics} studied similar questions for the best-response dynamics in games induced by envy-free pricing mechanisms. The liquid welfare has also been considered in other settings with budget constraints, including lottery pricing equilibria~\citep{Dughmi2016lottery}, online multi-unit auctions~\citep{Eden2019online}, and preferred deals for impression sales~\citep{Deng2019deals}.

\section{Preliminaries}\label{sec:prelim}
There are $m$ {\em items} to be allocated to $n$ {\em players}. Each player $i$ has a {\em valuation function} $v_i:2^m \rightarrow \mathbb{R}_{\geq 0}$ and a private {\em budget} $c_i \in \mathbb{R}_{\geq 0}$. The valuation function returns the value of player $i$ for every possible subset (or, {\em bundle}) of items\footnote{The players are assumed to have access to their valuation functions via demand queries; such a query takes as input a bundle of items and returns the value of the player for this bundle.} that can be given to her, while the budget restricts the payments that she can afford. To shorten our notation in the case of singleton bundles, we will write $v_i(j)$ instead of $v_i(\{j\})$. 
We assume that the valuation functions are monotone and subadditive. 
Let $S$ and $T$ be two bundles of items.
A valuation function $v_i$ is 
\begin{itemize}
\item {\em monotone} if $S \subseteq T$ implies that $v_i(S) \leq v_i(T)$, 
\item {\em additive} if $v_i(S \cup T) = v_i(S) + v_i(T)$, 
\item {\em fractionally-subadditive} if there exists a family of additive functions $\{\nu_1, ..., \nu_k\}$ such that $v_i(S) = \max_{\ell}\{\nu_\ell(S)\}$, and
\item {\em subadditive} if $v_i(S \cup T) \leq v_i(S) + v_i(T)$.
\end{itemize}
The class of additive valuation functions is a subset of the class of fractionally-subadditive functions, which in turn is a subset of the class of subadditive functions. Therefore, in this paper we consider the most general class of structured valuation functions.

We consider a class of simple combinatorial auctions $\calC$ for the allocation of the items to the players. Every player $i$ submits a vector of bids $\bb_i = (b_{ij})_{j \in [m]}$ consisting of a bid $b_{ij} \in \mathbb{R}_{\geq 0}$ per item $j \in [m]$; let $\bb = (\bb_i)_{i \in [n]}$ be the matrix consisting of the bids of all players for all items, and let $\bb_{[j]}= (b_{ij})_{i \in [n]}$ be the vector consisting of the bids of all players only for item $j$. Every auction in $\calC$ allocates independently each item $j$ to the player $\pi(j)$ that submits the highest bid for it ($b_{\pi(j),j} \geq b_{ij}$ for every $i \in [n]$), and also requests a payment of $p_j(\bb_{[j]}) \in [0, b_{\pi(j),j}]$ from player $\pi(j)$. \footnote{In case there are two or more players with the highest bid for an item, then such an item is arbitrarily given to one of these players; this captures the case where all bids for an item are equal to zero.} We assume that the payment function is non-decreasing and continuous in every bid.

\begin{example}
The two most prominent members of $\calC$ are the {\em simultaneous first and second price auctions} (SFPA and SSPA). 
In SFPA the winner of each item has to pay her own bid, while in SSPA the winner has to pay the second highest bid for the respective item. Both of these auctions clearly satisfy our assumptions about the payment functions: the payment is at most the bid of the winner, and increasing the bids cannot decrease the price. In fact, any auction that defines the payment to be a convex combination of the bids is a member of $\calC$.
\qed
\end{example}

Let $\XX(\bb)$ denote the {\em allocation} that is the outcome of the auction when given as input the bid matrix $\bb$, according to which player $i$ gets a bundle of items $X_i$ such that $\cup_i X_i = [m]$. The {\em utility} of player $i$ for the outcome of the auction is then defined as 
$$u_i(\bb) = v_i(X_i) - \sum_{j \in X_i}p_j(\bb_{[j]})$$
if her total payment does not exceed $c_i$, or $-\infty$ otherwise.

The players are utility-maximizers and strategically select the bids they submit to the auction. This strategic behavior defines a game $\calG$ among the players. We say that a bid matrix $\bb$ is a {\em pure Nash equilibrium} (or, simply, {\em equilibrium}) if no player has incentive to deviate to a different bid strategy, that is,  
$$u_i(\bb) \geq u_i(\yy,\bb_{-i}),$$ for every player $i$ and bid vector $\yy \neq \bb_i$. We use the notation $(\yy,\bb_{-i})$ to denote the matrix that is obtained after replacing the entries of $\bb$ corresponding to $\bb_i$ by $\yy$. Let $\eq(\calG)$ be the set of all equilibrium bid matrices of game $\calG$. 

The {\em liquid welfare} of an allocation $\XX$ is the total value of all players, such that the value of each player is capped off by her budget, that is, 
$$\LW(\XX) = \sum_{i \in [n]} \min\{v_i(X_i),c_i\}.$$
The {\em pure liquid price of anarchy} of an auction is the worst-case ratio (over all possible games induced by the auction) between the optimal liquid welfare achieved by any allocation and the {\em minimum} liquid welfare at equilibrium:
\begin{align*}
\lpoa &=  \sup_\calG \frac{\max_{\YY} \LW(\YY)}{\min_{\bb \in \eq(\calG)} \LW(\XX(\bb))}.
\end{align*}
Similarly, the {\em pure liquid price of stability} of an auction is the worst-case ratio (over all possible games induced by the auction) between the optimal liquid welfare achieved by any allocation and the {\em maximum} liquid welfare at equilibrium:
\begin{align*}
\lpos &=  \sup_\calG \frac{\max_{\YY} \LW(\YY)}{\max_{\bb \in \eq(\calG)} \LW(\XX(\bb))}. 
\end{align*}
Clearly, $\lpoa \geq \lpos \geq 1$, and hence any upper bound on the $\lpoa$ is also an upper bound on the $\lpos$, while any lower bound on the $\lpos$ is also a lower bound on the $\lpoa$. Furthermore, observe that an upper bound on the $\lpoa$ for the class of subadditive valuation functions is also an upper bound on the $\lpoa$  for any hierarchically lower class of valuation functions (additive and fractionally-subadditive), while a lower bound on the $\lpoa$ for the class of additive valuation functions is also a lower bound on the $\lpoa$ for any hierarchically higher class of valuation functions (fractionally-subadditive and subadditive). These hold for the $\lpos$ as well.

We finally assume that the players are {\em rational} and behave {\em conservatively} to avoid getting negative utility: the sum of bids they submit for any bundle of items does not exceed their value for this bundle nor their budget, that is, 
$$\sum_{j \in X_i} b_{ij} \leq \min\{v_i(X_i),c_i\}.$$ 
Besides being a natural assumption that has been extensively considered in the related literature, this restriction is also necessary in order to have meaningful bounds on the liquid price of anarchy. 

\begin{example}
In SSPA, it might happen that players with very small value or budget submit bids that are extremely large, and end up being the high bidders for some items. This may further lead to an equilibrium where the other players are discouraged enough to submit bids of zero for these items (even though they may have high value for them, as well as budget), leading to a payment of $0$, but also to an allocation with very low liquid welfare. 
Observe that such situations cannot appear in SFPA since the players have to pay their own bids in case they win, and are therefore conservative by definition. \qed
\end{example}

\section{Upper bounds for auctions in $\calC$}\label{sec:upper}
Here, we will show an upper bound of $2$ on the liquid price of anarchy of every auction in $\calC$. 
We start by focusing only on one arbitrary player $i$ and a given set of items $S$. By considering the unilateral deviation of $i$ to a particularly defined bid strategy over $S$, we prove two different bounds on the overall contribution of $i$ to the liquid welfare of the allocation that assigns $S$ to $i$, when the strategy leads to a total payment within $i$'s budget and when the payment exceeds the budget. These bounds will be used later to show the upper bound of $2$ (Theorem~\ref{thm:upper}), by partitioning the players into two sets depending on their contribution to the liquid welfare of the equilibrium. The proof of the next lemma carefully exploits ideas developed by \citet{Bhawalkar2011bidding} and \citet{Feldman2013simultaneous} to bound the price of anarchy for the case of SFPA and SSPA without budget constraints. 

\begin{lemma}\label{lem:deviation}
Consider a player $i$ with a subadditive valuation function $v_i$ and budget $c_i$. Let $S$ be a set of items, and $\bb_{-i}$ be the bid matrix containing the bids of the other players $\ell \neq i$ for all items. Then, there exists a bid vector $\yy_i = (y_{ij})_{j \in [m]}$ such that the utility of player $i$ when she unilaterally deviates to $\yy_i$ is 
$$u_i(\yy_i,\bb_{-i}) \geq \min\{v_i(S),c_i\} - \sum_{j \in S}\max_{\ell \neq i}b_{\ell j},$$
if her total payment does not exceed $c_i$.
Otherwise,
$$\min\{v_i(S),c_i\} < \sum_{j \in S} \max_{\ell \neq i}b_{\ell j}.$$
\end{lemma}

\begin{proof}
Let $T$ be a maximal subset of $S$ such that $v_i(T) \leq \sum_{j \in T} \max_{\ell \neq i} b_{\ell j}$. We define the deviating bid vector $\yy_i = (y_{ij})_{j \in [m]}$ to be such that
\begin{align*}
y_{ij} = 
\begin{cases}
\max_{\ell \neq i}b_{\ell j} + \delta, & \text{if } j \in S \setminus T \\
0, & \text{otherwise},
\end{cases}
\end{align*}
for any $\delta > 0$. Since the payment function of the auction is continuous, player $i$ definitely wins every item $j \in S \setminus T$, even for values of $\delta$ that are infinitesimally close to zero. Hence, in the rest of the proof, without loss of generality, we set $\delta=0$. 

Before we continue, we argue that for any set $Q \subseteq S\setminus T$, $\sum_{j \in Q} y_{ij} \leq v_i(Q)$, and hence $\yy_i$ is a conservative strategy for player $i$ in terms of her valuation function. Assume otherwise that there exists a non-empty set $Q^* \subseteq S\setminus T$ with $\sum_{j \in Q^*}y_{ij} > v(Q^*)$. Then, by the subadditivity of the valuation function, the definition of $\yy_i$ and the definition of $T$, we have that
\begin{align*}
v_i(Q^* \cup T) &\leq v_i(Q^*) + v_i(T) \leq \sum_{j \in Q^* \cup T} \max_{\ell \neq i}b_{\ell j}, 
\end{align*}
which contradicts the maximality of $T$. 

Let $\yy=(\yy_i,\bb_{-i})$, and denote by $\YY=\YY(\yy)$ the allocation after the deviation of player $i$ to the bid vector $\yy_i$. 
Since $p_j(\yy_{[j]}) \leq y_{ij}$ for every item $j \in Y_i$, $y_{ij} = 0$ for every $j \in Y_i \setminus (S\setminus T)$, and $y_{ij} = \max_{\ell \neq i}b_{\ell j}$ for every $j \in S\setminus T$, we can bound the total payment of player $i$ as follows:
\begin{align}\label{eq:payment}
 \sum_{j \in Y_i} p_j(\yy_{[j]}) \leq \sum_{j \in Y_i} y_{ij}  = \sum_{j \in S \setminus T} \max_{\ell \neq i}b_{\ell j}.
\end{align}

Now, we distinguish between the following two cases. We first assume that the total payment of player $i$ is within her budget $c_i$.
The utility of player $i$ is 
\begin{align*}
u_i(\yy) = v_i(Y_i) - \sum_{j \in Y_i} p_j(\yy_{[j]}).
\end{align*}
Since $Y_i \supseteq S \setminus T$, by the monotonicity of $v_i$, we have that $v_i(Y_i) \geq v_i(S \setminus T)$. 
Combining this with inequality~\eqref{eq:payment}, we obtain
\begin{align*}
u_i(\yy) \geq v_i(S \setminus T) - \sum_{j \in S\setminus T} \max_{\ell \neq i} b_{\ell j}.
\end{align*}
By the definition of $T$, we have $v_i(T) - \sum_{j \in T} \max_{\ell \neq i} b_{\ell j} \leq 0$, and hence
\begin{align*}
u_i(\yy) &\geq v_i(S \setminus T) - \sum_{j \in S \setminus T} \max_{\ell \neq i} b_{\ell j}
+ v_i(T) - \sum_{j \in T} \max_{\ell \neq i} b_{\ell j} \\
&= v_i(S \setminus T) + v_i(T) - \sum_{j \in S} \max_{\ell \neq i} b_{\ell j}.
\end{align*}
By the subadditivity of $v_i$, we further have that
\begin{align*}
v_i(S \setminus T) + v_i(T) \geq v_i( (S\setminus T)\cup T) = v_i(S).
\end{align*}
By also using the simple fact that $\alpha \geq \min\{\alpha,\beta\}$ for every $\alpha,\beta$, we finally obtain the desired inequality:
\begin{align*}
u_i(\yy) &\geq v_i(S) - \sum_{j \in S} \max_{\ell \neq i} b_{\ell j}\\
&\geq \min\{v_i(S),c_i\} - \sum_{j \in S}\max_{\ell \neq i}b_{\ell j}.
\end{align*}

We now switch to the second case where the total payment of player $i$ exceeds her budget. Then, by the fact that $\alpha \geq \min\{\alpha,\beta\}$ for every $\alpha,\beta$ and inequality~\eqref{eq:payment}, we have that
\begin{align*}
\min\{v_i(S),c_i\} &\leq c_i < \sum_{j \in Y_i} p_j(\yy_{[j]}) \leq \sum_{j \in S \setminus T} \max_{\ell \neq i}b_{\ell j} \leq \sum_{j \in S} \max_{\ell \neq i}b_{\ell j}.
\end{align*}
This completes the proof of the lemma. 
\end{proof}

Next, we show the main result of this section. The proof idea is to partition the liquid welfare at equilibrium into two quantities, one for the players that contribute their budget and one for the players that contribute their value, and bound each of them separately, using Lemma~\ref{lem:deviation} and the equilibrium condition.

\begin{theorem}\label{thm:upper}
The pure liquid price of anarchy of any auction in $\calC$ with subadditive valuation functions is at most $2$.
\end{theorem}

\begin{proof}
Consider an arbitrary instance with $n$ players and $m$ items, where player $i$ has a subadditive valuation function $v_i$ and a budget $c_i$. Let $\bb$ be an equilibrium bid matrix that induces the allocation $\XX = \XX(\bb)$ according to which player $i$ gets a set of items $X_i$. We also denote by $O_i$ the set given to player $i$ in an optimal allocation $\OO$.

Let $V = \{i : v_i(\XX) \leq c_i \}$ be the set of players that contribute their value to the liquid welfare at equilibrium. We can write the liquid welfare at equilibrium as 
\begin{align}\nonumber
\LW(\XX) &= \sum_{i \in [n]} \min\{v_i(\XX),c_i\} \\\label{eq:first-LW}
&=\sum_{i \not\in V} \min\{v_i(\XX),c_i\} + \sum_{i \in V}\min\{v_i(\XX),c_i\}.
\end{align}
In the rest of the proof we will work on these two quantities separately, with the goal of lower-bounding them in terms of (parts of) the optimal liquid welfare. 

For every player $i\not\in V$, we have that $\min\{v_i(\XX), c_i \} = c_i \geq \min\{v_i(O_i),c_i\}$, and therefore, by summing over all such players, we obtain
\begin{align}\label{eq:first-not-V}
\sum_{i \not\in V} \min\{v_i(\XX), c_i \} \geq \sum_{i \not\in V} \min\{v_i(O_i), c_i \}. 
\end{align}

Now, we focus on players in $V$. For such a player $i$, let $\yy_i$ be the bid vector of Lemma~\ref{lem:deviation} with $S=O_i$. By the definition of $V$ and the utility definition, we have that 
$$\min\{v_i(\XX),c_i\} = v_i(\XX) = u_i(\bb) + \sum_{j \in X_i} p_j(\yy_{[j]}) \geq u_i(\bb).$$

If the total payment of player $i$ does not exceed her budget $c_i$ when she deviates to $\yy_i$, by the equilibrium condition and Lemma~\ref{lem:deviation}, we have that
\begin{align*}
\min\{v_i(\XX),c_i\} \geq u_i(\bb) &\geq u_i(\yy_i,\bb_{-i}) \\
&\geq \min\{v_i(O_i),c_i\} - \sum_{j \in O_i}\max_{\ell \neq i}b_{\ell j}.
\end{align*}
Otherwise, we obtain exactly the same inequality by the fact that $u_i(\bb) \geq 0$, and $\min\{v_i(O_i),c_i\} - \sum_{j \in O_i}\max_{\ell \neq i}b_{\ell j} < 0$. By summing this inequality over all players in $V$, we get
\begin{align}\label{eq:V1}
&\sum_{i \in V}\min\{v_i(\XX),c_i\} 
\geq \sum_{i\in V} \min\{v_i(O_i),c_i\}  - \sum_{i \in V} \sum_{j \in O_i}\max_{\ell \neq i}b_{\ell j}.
\end{align}
Further, we have 
\begin{align*}
\sum_{i \in V} \sum_{j \in O_i}\max_{\ell \neq i}b_{\ell j} &\leq \sum_{i \in [n]} \sum_{j \in O_i} \max_{\ell \in [n]} b_{\ell j} \\
&= \sum_{j \in [m]} \max_{\ell \in [n]} b_{\ell j} = \sum_{i \in [n]} \sum_{j \in X_i} b_{ij} \leq \LW(\XX),
\end{align*}
where the two equalities follow since both $\OO$ and $\XX$ define (possibly different) partitions of the set of items, while the last inequality follows by the definition of the liquid welfare, and by the conservative behavior of player $i$ at equilibrium, which requires that $\sum_{j \in X_i} b_{ij} \leq \min\{v_i(X_i), c_i\}$. By this, \eqref{eq:V1} now becomes
\begin{align}\label{eq:first-V} 
&\sum_{i \in V}\min\{v_i(\XX),c_i\}
\geq \sum_{i\in V} \min\{v_i(O_i),c_i\} - \LW(\XX).
\end{align}

By substituting \eqref{eq:first-not-V} and \eqref{eq:first-V} into \eqref{eq:first-LW}, we get
\begin{align*}
\LW(\XX) &\geq \sum_i \min\{v_i(O_i),c_i\} - \LW(\XX),
\end{align*}
which implies that
\begin{align*}
\lpoa = \frac{\LW(\OO)}{\LW(\XX)} \leq 2,
\end{align*}
and the proof is now complete.
\end{proof}

Since SFPA and SSPA are members of $\calC$, we immediately obtain the following corollary.

\begin{corollary}\label{cor:first-second}
The pure liquid price of anarchy of SFPA and SSPA with subadditive valuation functions is at most $2$.
\end{corollary}

\section{Lower bounds for auctions in $\calC$}\label{sec:lower}
In this section, we will present lower bounds on the price of anarchy and the price of stability of the simple combinatorial auctions we consider in this paper. 
Before we dive into the results of this section, we remark that the bound of $2$ is {\em tight} for the pure $\lpoa$ of SFPA and SSPA with subadditive valuation functions. This follows by the upper bound of Corollary~\ref{cor:first-second} and the corresponding lower bound presented by \citet{Azar2017lpoa} for the case of additive valuation functions. This extends the tight bound of \citet{Azar2017lpoa} for fractionally-subadditive functions.

Using a simplified version of the instance considered by \citet{Azar2017lpoa}, we will show with the next theorem that the lower bound of $2$ actually holds for the pure liquid price of stability, and for any simple combinatorial auction with a payment function that is a convex combination of the bids; to simplify our discussion, we refer to such auctions as {\em convex auctions}. 

\begin{theorem}\label{thm:lower-convex}
The pure liquid price of stability of any convex auction in $\calC$ with additive valuation functions is at least $2-\varepsilon$, for any arbitrarily small constant $\varepsilon >0$.
\end{theorem}

\begin{proof}
Consider an instance with two players and two items. Both players have additive valuation functions so that their values for the individual items are $v_1(1)=1$, $v_1(2)=1$, $v_2(1)=0$, and $v_2(2)=1$, while their budgets are $c_1=1$ and $c_2=1-\varepsilon$.

Since $v_2(1)=0$, player $2$ has no incentive to bid more than $0$ for item $1$. Hence, her main objective is to acquire item $2$, by bidding $b_{22} \in [0, 1-\varepsilon]$. Now, let $\gamma$ and $\delta$ be two arbitrarily small but strictly positive constants such that $\gamma + \delta < \varepsilon$. Player $1$ can acquire item $1$ almost for free by bidding $b_{11}=\gamma$, and item $2$ by bidding $b_{12} = b_{22} + \delta$. This is a valid strategy for player $1$ since $b_{11} + b_{12} = \gamma + b_{22} + \delta < \gamma + 1 -\varepsilon + \delta < 1$, which is at most her value for both items $v_1(1)+v_1(2)=2$, and at most her budget $c_1=1$.
Since the payment of player $1$ is at most $b_{11} + b_{12} < 1$ for any simple combinatorial auction with a payment function that is a convex combination of the bids of the players, her utility is at least $2-b_{11}-b_{12} > 1$. This is the maximum possible utility that player $1$ can obtain since in any other allocation she gets at most one item (and thus her utility is at most $1$). 

Consequently, $\XX = (\{1,2\},\varnothing)$ is the only possible equilibrium allocation, with 
$$\LW(\XX) = \min\{v_1(\{1,2\}),c_1\} + \min\{v_2(\varnothing),c_2\}  = 1.$$
On the other hand, the optimal allocation is $\OO=(\{1\},\{2\})$ with 
$$\LW(\OO) = \min\{v_1(1),c_1\} + \min\{v_2(2),c_2\}= 2-\varepsilon,$$ 
and the lower bound on the $\lpos$ follows. 
\end{proof}

By Theorem~\ref{thm:upper}, and since additive valuation functions form a subclass of subadditive valuation functions, we also obtain the following corollary.

\begin{corollary}\label{cor:convex}
The pure price of stability of any convex auction in $\calC$ with subadditive valuation functions is at most $2$, and this bound is tight.
\end{corollary}

We now consider the whole class of simple combinatorial auctions and show a weaker lower bound on the liquid price of anarchy, which depends on both the number of players and the number of items. The main idea of the proof is to construct two instances with different private information (valuation functions and budgets) for the players, which are such that the strategic behavior of the players leads to the same equilibrium in both instances. Consequently, the auction cannot tell the two instances apart, and fails to learn the private information of the players. 

\begin{theorem}\label{thm:lower-general}
The pure liquid price of anarchy of any auction in $\calC$ with subadditive valuation functions is at least $2-\frac{n-1}{m}-\frac{1}{n}$.
\end{theorem}

\begin{proof}
Consider some arbitrary auction in $\calC$, and an instance with $n$ players and $m \geq \lambda n$ items, for some $\lambda \geq 2$. Let $v: 2^m \rightarrow \mathbb{R}_{\geq 0}$ be any additive (and, thus, subadditive) function over the powerset of items such that $v(S)=\sum_{j \in S}v(j)$ for every $S \subseteq [m]$, and also let $V = \sum_{j \in [m]} v(j)$. Every player $i \in [n]$ has a valuation function $v_i = v$ and budget $c_i=+\infty$. Let $\bb$ be an equilibrium bid matrix of the induced game, which leads to an allocation $\XX$ according to which player $i$ gets a set of items $X_i \subseteq [m]$. Since the game is symmetric (all players have the same valuation function and budget), any allocation has the same liquid welfare, and hence the liquid price of anarchy of this game is equal to $1$. 

Next, we consider a second instance with the same set of players and items, but with modified valuation functions and budgets. Let $\ell = \arg\min_i v(X_i)$ be the player that gets the least value at the equilibrium $\bb$ of the previous instance; since $v$ is an additive function and $V = \sum_{j \in [m]} v(j)$, it must be $v(X_\ell) \leq V/n$.
In the new instance, player $\ell$ has the same valuation function $v_\ell=v$ and budget $c_\ell$ as before. In contrast, every other player $i \neq \ell$ has a modified valuation function $\tilde{v}_i$ such that $\tilde{v}_i(\varnothing)=0$ and $\tilde{v}_i(S) = v(S)+v(X_i)$ for every $S \neq \varnothing$, and a modified budget $\tilde{c}_i=v(X_i)$. Clearly, the valuation function $\tilde{v}_i$ is subadditive: For any two non-empty sets of items $S$ and $T$, by the definition of $\tilde{v}$, and the properties of $v$ (subadditivity and monotonicity), we have that 
\begin{align*}
\tilde{v}_i(S \cup T) &= v(S \cup T) + v(X_i) \leq v(S) + v(T) + v(X_i) \leq \tilde{v}_i(S) + \tilde{v}_i(T).
\end{align*}
Observe that for any bid matrix $\yy$, the utility of player $\ell$ in the second instance is $\tilde{u}_\ell(\yy)=u_\ell(\yy)$ and the utility of any other player $i\neq \ell$ is $\tilde{u}_i(\yy) = u_i(\yy) + v(X_i)$. Hence, since the term $v(X_i)$ does not affect the way that player $i$ selects her strategy (it is viewed as a constant), $\bb$ must be an equilibrium bid matrix for the second instance as well, which leads to the same allocation $\XX$, according to which each player $i$ gets the set $X_i$; observe that the payment of player $i \neq \ell$ is within her budget $\tilde{c}_i = v(X_i)$ since she has non-negative utility at the equilibrium of the first instance, meaning that $v(X_i) - \sum_{j \in X_i} p_j(\bb_{[j]}) \geq 0$.
Consequently, the liquid welfare at equilibrium is
\begin{align*}
\LW(\XX) &= \min\{v_\ell(X_\ell),c_\ell\} + \sum_{i \neq \ell} \min\{\tilde{v}_i(X_i),\tilde{c}_i\} \\
&= \min\{v(X_\ell),c_\ell\} + \sum_{i \neq \ell} \min\{2v(X_i),v(X_i)\} \\
&= \sum_{i \in [n]} v(X_i) = V.
\end{align*}
On the other hand, consider the allocation $\OO$ according to which player $\ell$ gets the $m-n+1$ most valuable items according to the function $v$, and every player $i \neq \ell$ is given one of the remaining items; note that, since $m \geq \lambda n$, there are enough items to define such an allocation. 
The liquid welfare of $\OO$ is 
\begin{align}\nonumber
\LW(\OO) &= \min\{v_\ell(O_\ell),c_\ell\} + \sum_{i \neq \ell} \min\{\tilde{v}_i(O_i),\tilde{c}_i\} \\ \nonumber
&= \min\{v(O_\ell),+\infty\} + \sum_{i \neq \ell} \min\{v(O_i)+v(X_i),v(X_i)\} \\ \label{eq:OPT-lower}
&= v(O_\ell) + \sum_{i \neq \ell} v(X_i).
\end{align}
Since $v(O_\ell) \geq (m-n+1) \cdot v(j)$ for each $j \in [m]\setminus O_\ell$ and $|[m]\setminus O_\ell|=n-1$, we have that $v([m]\setminus O_\ell) \leq \frac{n-1}{m-n+1} \cdot v(O_\ell)$. By this, we obtain 
\begin{align*}
&V = v(O_\ell) + v([m]\setminus O_\ell) \leq \left( 1 + \frac{n-1}{m-n+1}\right) v(O_\ell) \\
&\Leftrightarrow  v(O_\ell) \geq \left(1 - \frac{n-1}{m}\right) \cdot V.
\end{align*}
Moreover, by the definition of $V$ and since $v(X_\ell) \leq V/n$, we have
$$\sum_{i \neq \ell} v(X_i) = V - v(X_\ell) \geq \left(1 - \frac{1}{n}\right) V.$$ 
Substituting the last two expressions into \eqref{eq:OPT-lower}, we get
\begin{align*}
\LW(\OO) &\geq v(O_\ell) + \sum_{i \neq \ell} v(X_i) \geq \left(2 - \frac{n-1}{m} - \frac{1}{n} \right) V. 
\end{align*}
Therefore, the liquid price of anarchy is at least
\begin{align*}
\frac{\LW(\OO)}{\LW(\XX)} \geq 2 - \frac{n-1}{m} -\frac{1}{n}.
\end{align*}
as desired.
\end{proof}

By considering instances in which the number of items is large (tends to infinity), we recover the bound of $2-1/n$ proved by \citet{CV18} for single divisible resource allocation mechanisms. Essentially, we can interpret Theorem~\ref{thm:lower-general} as a discrete version of the corresponding theorem of \citet{CV18}, which shows how the bound depends not only on the number of players, but also on the number of items. Consequently, it leaves room for improvement for special cases where the number of items is bounded.

\section{More complex auctions}\label{sec:vcg}
So far, we have focused on auctions that greedily allocate the items, by simply looking at who submits the highest bid for every individual item, and showed that no such auction can achieve full efficiency. Naturally, one might wonder whether taking into account the bids over all items while coming up with the allocation can improve the liquid price of anarchy and stability. In this section, we answer this negatively for the well-known VCG auction~\citep{V61,C71,G73}.  

Let us first give a brief description of the auction. Given as input a matrix $\bb = (b_i(S))_{i \in [n], S \in 2^{[m]}}$ that specifies the bid of every player $i$ for every possible bundle of items $S$, VCG computes an allocation $\XX$ that maximizes the social welfare according to $\bb$: 
$$\XX \in \arg\max_\YY \sum_{i\in [n]} b_i(\YY).$$
Then, the payment $p_i(\bb)$ of player $i$ who is allocated bundle $X_i$ is the difference between the maximum possible social welfare that could have been achieved if $i$ did not participate and the social welfare of all players besides $i$ for $\XX$, always according to the given bids:
$$p_i(\bb) = \max_\YY \sum_{\ell \neq i} b_\ell(\YY) - \sum_{\ell \neq i} b_\ell(\XX).$$

In the no-budget setting, VCG is known to be truthful (in the sense that it is a dominant strategy for every player to bid her true value for every possible bundle of items) and also achieves full efficiency. However, when the players have budget constraints, VCG is no longer truthful~\citep{Dobzinski2014liquid}, and as we will show with the next theorem, it is not fully efficient either. 

\begin{theorem}
The pure liquid price of stability of VCG with additive valuation functions is at least $2-\varepsilon$, for any arbitrarily small constant $\varepsilon >0$.
\end{theorem}

\begin{proof}
Let $\alpha < \varepsilon$ be a parameter, and consider an instance with two players and two items. The valuation functions of the players are additive so that their values for the individual items are $v_1(1)=1$, $v_1(2)=1-\alpha$, $v_2(1)=0$, and $v_2(2)=1$, while their budgets are $c_1=1$ and $c_2=1-\varepsilon$.

Before we argue about the liquid price of stability, observe that if both players were truthful, then player $1$ would get item $1$ and player $2$ would get item $2$ as this is the allocation that maximizes the social welfare. However, the payment of player $2$ would be equal to $1-\alpha$, which is the difference between the value of player $1$ for both items and her value for only item $1$. Since $\alpha < \varepsilon$, this means that player $2$ would need to pay an amount that exceeds her budget $c_2=1-\varepsilon$. Hence, she would get utility $-\infty$ and would want to deviate to a smaller bid for item $2$.

Similarly to the proof of Theorem~\ref{thm:lower-convex}, we can argue that the only possible equilibrium allocation is $\XX = (\{1,2\},\varnothing)$. To see this, observe that since the value of player $2$ is $0$ for item $1$, she will bid $b_2 \in [0,1-\varepsilon]$ for the singleton set $\{2\}$ and the set $\{1,2\}$, and $0$ for the singleton set $\{1\}$. Now, player $1$ can bid $0$ for both singleton sets $\{1\}$ and $\{2\}$, and $b_1 = 1 > b_2$ for $\{1,2\}$. Consequently, the social welfare is maximized by allocating both items to player $1$, who has to pay $b_2$. Since the utility of player $1$ for any other allocation is at most $1$ (she gets at most one item), and her utility now is $2-\alpha - b_2 \geq 1+\varepsilon-\alpha > 1$, player $1$ maximizes her utility, $\XX$ is the only possible equilibrium allocation, and the bound on the $\lpos$ follows. 
\end{proof}

\section{Known budgets}\label{sec:known}
In this section, we no longer assume that budgets are private. Instead, we consider the case where the budgets are common knowledge, and the auction can take them into account while defining the payments of the winners. Still, even with this extra power, no simple combinatorial auction can achieve optimal liquid welfare. 
The proof of the following theorem is similar to that of Theorem~\ref{thm:lower-general}, but uses instances with two players and the same budgets, since these are now assumed to be known. Inevitably, this leads to a somewhat weaker bound. Similarly to Theorem~\ref{thm:lower-general}, the bound can be interpreted as a discrete version of the $4/3$ bound of \citet{CV18} for the class of the so-called budget-aware single divisible resource allocation mechanisms. 

\begin{theorem}\label{thm:lower-known}
The pure liquid price of anarchy of any auction in $\calC$ with subadditive valuation functions is at least $\frac{4}{3} - \frac{2}{3m}$, even when the budgets are known.
\end{theorem}

\begin{proof}
Consider some arbitrary auction in $\calC$, and an instance with two players and $m \geq 2$ items. Let $v: 2^m \rightarrow \mathbb{R}_{\geq 0}$ be any additive (and, thus, subadditive) function over the powerset of items such that $v(S)=\sum_{j \in S}v(j)$ for every $S \subseteq [m]$, and also let $V = \sum_{j \in [m]} v(j)$. The valuation functions of the two players are $v_1 = v_2= v$ and their budgets are $c_1=c_2=V$. Let $\bb$ be an equilibrium bid matrix of the induced game, which leads to an allocation $\XX$ according to which player $1$ gets set $X_1$ and player $2$ gets set $X_2$. Without loss of generality, assume that $v(X_1) \leq v(X_2)$, and hence $v(X_1) \leq V/2$. 
Since the game is symmetric (all players have the same valuation function and budget), all allocations have the same liquid welfare, and hence the liquid price of anarchy of this game is equal to $1$.

Next, we consider a second instance with the same set of players, items, and budgets, but with modified valuation functions. In the new instance, player $1$ has the same valuation function $v_1=v$ as in the previous instance, while player $2$ has a modified valuation function $\tilde{v}_2$ such that $\tilde{v}_2(\varnothing)=0$ and $\tilde{v}_2(S) = v(S)+V$ for every $S \neq \varnothing$; it should be obvious that $\tilde{v}_2$ is subadditive.

Since $V$ is only a constant term in the modified valuation function of player $2$, $\bb$ must be an equilibrium bid matrix for the second instance as well. This leads to the same allocation $\XX$, according to which player $1$ gets set $X_1$ and player $2$ gets set $X_2$; the payments of both players are within their budgets since this is true in the first instance and their budgets did not change. 
The liquid welfare at equilibrium is
\begin{align*}
\LW(\XX) &= \min\{v_1(X_1),c_1\} + \min\{\tilde{v}_2(X_2),c_2\} \\
&= \min\{v(X_1),V\} + \min\{v(X_2) + V,V\} \\
&= v(X_1) + V \\
&\leq \frac{3V}{2}.
\end{align*}
On the other hand, consider the allocation $\OO$ according to which player $1$ gets the $m-1$ most valuable items according to the function $v$, and player $2$ gets the remaining item. We clearly have that $v(O_2) \leq \frac{1}{m-1}v(O_1)$ and 
\begin{align*}
&V = v(O_1) + v(O_2) \leq \left(1 + \frac{1}{m-1}\right)v(O_1) \\
&\Leftrightarrow v(O_1) \geq \left(1 - \frac{1}{m}\right) V. 
\end{align*} 
The liquid welfare of $\OO$ is 
\begin{align*}
\LW(\OO) &= \min\{v_1(O_1),c_1\} + \min\{\tilde{v}_2(O_2),c_2\} \\ 
&= \min\{v(O_1),V\} + \min\{v(O_2)+V,V\} \\
&= v(O_1) + V \\
&\geq \left(2 - \frac{1}{m}\right) V.
\end{align*}
Therefore, the liquid price of anarchy is at least
\begin{align*}
\frac{\LW(\OO)}{\LW(\XX)} \geq \frac{4}{3} - \frac{2}{3m},
\end{align*}
as desired. 
\end{proof}

\section{Conclusions and possible extensions}\label{sec:open}
In this paper, we studied the efficiency of a class of simple combinatorial auctions, which allocate each item separately to strategic players with subadditive valuation functions and budgets. 
We showed tight bounds on the $\lpoa$ and $\lpos$ for convex simple auctions, and complemented it with (weaker) lower bounds for every simple auction, for the more complex VCG auction, and also for simple auctions that may have access to the budgets of the players. Even though we painted an almost complete picture of the $\lpoa$ and $\lpos$ landscape for these auctions, there are still many interesting directions to be explored. 

In terms of our results, the lower bound of $2-1/n-(n-1)/m$ for private budgets (Theorem~\ref{thm:lower-general}) and the lower bound of $4/3 - 2/(3m)$ for known budgets (Theorem~\ref{thm:lower-known}), leave open the possibility of simple combinatorial auctions with efficiency guarantees that are strictly better than the upper bound of $2$ which we proved in Theorem~\ref{thm:upper} for special cases, depending on the number of players or items. In particular, is there a better simple auction for the fundamental case where there are only two players? What about a constant number of players and known budgets?

Furthermore, in this paper we focused exclusively on the case of pure equilibria. \citet{Azar2017lpoa} did focus on more general equilibrium concepts, but were only able to prove non-tight constant bounds on the liquid price of anarchy for SFPA and SSPA with agents that have additive valuation functions. Consequently, an important avenue for future research is to consider more general equilibrium concepts and valuation functions, and bound the liquid price of anarchy and stability for a broader class of simple auctions. More concretely, what happens in the case of convex auctions with (fractionally-)subadditive agents and Bayes-Nash equilibria? 

\bibliographystyle{named}
\bibliography{bib-lpoa}

\end{document}